\documentclass[pra,nopacs,twocolumn,10pt]{revtex4}
\usepackage{amsmath}
\usepackage{graphicx,,epsfig,amsmath,latexsym,amssymb,verbatim,color}
\usepackage{dsfont}
\usepackage{theorem}
\usepackage{url}

 \usepackage{hyperref}
\hypersetup{colorlinks=true,citecolor=blue,linkcolor=blue,filecolor=blue,urlcolor=blue,breaklinks=true}

\newtheorem{definition}{Definition}
\newtheorem{proposition}[definition]{Proposition}

\newtheorem{theorem}[definition]{Theorem}

\def\squareforqed{\hbox{\rlap{$\sqcap$}$\sqcup$}}
\def\qed{\ifmmode\squareforqed\else{\unskip\nobreak\hfil
\penalty50\hskip1em\null\nobreak\hfil\squareforqed
\parfillskip=0pt\finalhyphendemerits=0\endgraf}\fi}
\def\endenv{\ifmmode\;\else{\unskip\nobreak\hfil
\penalty50\hskip1em\null\nobreak\hfil\;
\parfillskip=0pt\finalhyphendemerits=0\endgraf}\fi}
\newenvironment{proof}{\noindent \textbf{{Proof~} }}{\qed}
\newenvironment{remark}{\noindent \textbf{{Remark~}}}{\qed}
\newenvironment{example}{\noindent \textbf{{Example~}}}{\qed}

\mathchardef\ordinarycolon\mathcode`\:
\mathcode`\:=\string"8000
\def\vcentcolon{\mathrel{\mathop\ordinarycolon}}
\begingroup \catcode`\:=\active
  \lowercase{\endgroup
  \let :\vcentcolon
  }

\newcommand{\nc}{\newcommand}
\nc{\rnc}{\renewcommand}
\nc{\beg}{\begin{equation}}
\nc{\eeq}{{\end{equation}}}
\nc{\beqa}{\begin{eqnarray}}
\nc{\eeqa}{\end{eqnarray}}
\nc{\lbar}[1]{\overline{#1}}
\nc{\bra}[1]{\langle#1|}
\nc{\ket}[1]{|#1\rangle}
\nc{\ketbra}[2]{|#1\rangle\!\langle#2|}
\nc{\braket}[2]{\langle#1|#2\rangle}

\nc{\proj}[1]{| #1\rangle\!\langle #1 |}
\nc{\avg}[1]{\langle#1\rangle}
\nc{\Rank}{\operatorname{Rank}}
\nc{\smfrac}[2]{\mbox{$\frac{#1}{#2}$}}
\nc{\tr}{\operatorname{Tr}}
\nc{\ox}{\otimes}
\nc{\dg}{\dagger}
\nc{\dn}{\downarrow}
\nc{\cA}{{\cal A}}
\nc{\cB}{{\cal B}}
\nc{\cC}{{\cal C}}
\nc{\cD}{{\cal D}}
\nc{\cE}{{\cal E}}
\nc{\cF}{{\cal F}}
\nc{\cG}{{\cal G}}
\nc{\cH}{{\cal H}}
\nc{\cI}{{\cal I}}
\nc{\cJ}{{\cal J}}
\nc{\cK}{{\cal K}}
\nc{\cL}{{\cal L}}
\nc{\cM}{{\cal M}}
\nc{\cN}{{\cal N}}
\nc{\cO}{{\cal O}}
\nc{\cP}{{\cal P}}
\nc{\cQ}{{\cal Q}}
\nc{\cR}{{\cal R}}
\nc{\cS}{{\cal S}}
\nc{\cT}{{\cal T}}
\nc{\cX}{{\cal X}}
\nc{\cY}{{\cal Y}}
\nc{\cZ}{{\cal Z}}
\nc{\cW}{{\cal W}}
\nc{\csupp}{{\operatorname{csupp}}}
\nc{\qsupp}{{\operatorname{qsupp}}}
\nc{\var}{{\operatorname{var}}}
\nc{\rar}{\rightarrow}
\nc{\lrar}{\longrightarrow}
\nc{\polylog}{{\operatorname{polylog}}}
\nc{\1}{{\mathds{1}}}
\nc{\wt}{{\operatorname{wt}}}
\nc{\av}[1]{{\left\langle {#1} \right\rangle}}
\nc{\supp}{{\operatorname{supp}}}

\def\a{\alpha}

\nc{\RR}{{{\mathbb R}}}
\nc{\CC}{{{\mathbb C}}}
\nc{\FF}{{{\mathbb F}}}
\nc{\NN}{{{\mathbb N}}}
\nc{\ZZ}{{{\mathbb Z}}}
\nc{\PP}{{{\mathbb P}}}
\nc{\QQ}{{{\mathbb Q}}}
\nc{\UU}{{{\mathbb U}}}
\nc{\EE}{{{\mathbb E}}}
\nc{\id}{{\operatorname{id}}}

\nc{\CHSH}{{\operatorname{CHSH}}}

\nc{\be}{\begin{equation}}
\nc{\ee}{{\end{equation}}}
\nc{\bea}{\begin{eqnarray}}
\nc{\eea}{\end{eqnarray}}
\nc{\<}{\langle}
\rnc{\>}{\rangle}
\nc{\Hom}[2]{\mbox{Hom}(\CC^{#1},\CC^{#2})}
\nc{\rU}{\mbox{U}}

\nc{\ob}[1]{#1}

\nc{\SEP}{{\text{SEP}}}
\nc{\NS}{{\text{NS}}}
\nc{\LOCC}{{\text{LOCC}}}
\nc{\PPT}{{\text{PPT}}}
\nc{\EXT}{{\text{EXT}}}
\nc{\Sym}{{\operatorname{Sym}}}

\nc{\ERLO}{{E_{\text{r,LO}}}}
\nc{\ERLOCC}{{E_{\text{r,LOCC}}}}
\nc{\ERPPT}{{E_{\text{r,PPT}}}}
\nc{\ERLOCCinfty}{{E^{\infty}_{\text{r,LOCC}}}}
\nc{\Aram}{{\operatorname{\sf A}}}

\begin{document}
\title{Nonadditivity of Rains' bound for distillable entanglement}
\author{Xin Wang$^{1}$}
\email{xin.wang-8@student.uts.edu.au}
\author{Runyao Duan$^{1,2}$}
\email{runyao.duan@uts.edu.au}

\affiliation{$^1$Centre for Quantum Software and Information, Faculty of Engineering and Information Technology, University of Technology Sydney, NSW 2007, Australia}
\affiliation{$^2$UTS-AMSS Joint Research Laboratory for Quantum Computation and Quantum Information Processing, Academy of Mathematics and Systems Science, Chinese Academy of Sciences, Beijing 100190, China}
\begin{abstract}
Rains' bound is arguably the best known upper bound of the distillable entanglement by operations completely preserving positivity of partial transpose (PPT) and was conjectured to be additive and coincide with the asymptotic relative entropy of entanglement. We disprove both conjectures by explicitly constructing a special class of mixed two-qubit states. We then introduce an additive semidefinite programming  lower bound ($E_M$) for the asymptotic Rains' bound, and it immediately becomes a computable lower bound for entanglement cost of bipartite states. Furthermore, $E_M$ is also proved to be the best known upper bound of the PPT-assisted deterministic distillable entanglement and gives the asymptotic rates for all pure states and some class of genuinely mixed states. 
\end{abstract}
\maketitle

\section{Introduction} 
Entanglement plays a crucial role in quantum physics and is the key resource in many quantum information processing tasks.  So it is quite natural and important to develop a theoretical framework to describe and quantify it. In spite of a series of remarkable recent progress in the theory of entanglement (for reviews see, e.g., Refs. \cite{Plenio2007, Eisert2006, Christandl2006, Horodecki2009a}),  many fundamental questions or challenges still remain open. One of the most significant questions is to determine the \emph{distillable entanglement} $E_D$, i.e. the highest rate at which one can obtain maximally entangled states from an entangled state by local operations and classical communication (LOCC) \cite{Bennett1996c, Rains1999a}. This fundamental measure fully captures the ability of given state shared between distant parties (Alice and Bob) to generate strongly correlated qubits in order to allow reliable quantum teleportation or quantum cryptography. However, up to now, how to calculate $E_D$ for general quantum states still remains unknown. Also, in many practical applications, the resources are finite and the number of prepared states is limited. It is also of importance to study the deterministic distillable entanglement of finite entanglement transformations.

To evaluate the distillation rates efficiently, one possible way is to find computable upper bounds. A well-known upper bound of the distillable entanglement is the relative entropy of entanglement (REE) \cite{Vedral1997, Vedral1998a, Vedral1997a}, which expresses the minimal distinguishability between the given state and all possible separable states. An improved bound is the Rains' bound \cite{Rains2001},  which is arguably the best known upper bound of distillable entanglement. The best known SDP upper bound is introduced in Ref. \cite{Wang2016} and it is an improved version of the \emph{logarithmic negativity} \cite{Vidal2002,Plenio2005b}. In Ref \cite{Plenio2005b},  the logarithmic negativity is proved to be a proper entanglement monotone for the first time.  Other known upper bounds of $E_D$ are studied in Refs. \cite{Vedral1998a, Rains1999, Horodecki2000a, Christandl2004}. Unfortunately, most of these known upper bounds are difficult to compute \cite{Huang2014} and usually easily computable only for states with high symmetries, such as Werner states, isotropic states, or the family of ``iso-Werner'' states \cite{Bennett1996c, Vollbrecht2001, Terhal2000, Rains1999}.

Entanglement cost $E_C$ \cite{Bennett1996c, Hayden2001} is another fundamental measure in entanglement theory, which quantifies the rate for converting maximally entangled states to the given state by LOCC alone.
However, computing $E_C$ is NP-hard  \cite{Huang2014} and the entanglement cost is known only for a few of quantum states \cite{Vidal2002b, Yura2003, Matsumoto2004}. Even under the PPT operations, there are only bounds for the exact entanglement cost \cite{Audenaert2003}. 

Since both distillable entanglement and the entanglement cost are important but difficult to compute, it is of great significance to find the best approach to efficiently evaluate them. As Rains' bound is proved to be equal to the asymptotic relative entropy of entanglement  (with respect to PPT states)  for Werner states \cite{Audenaert2001} and orthogonally invariant states \cite{Audenaert2002}, one open problem is to determine whether these two quantities always coincide \cite{Plenio2007}. Another significant open problem is whether Rains' bound is additive, and it was conjectured in Ref. \cite{Audenaert2002} that Rains' bound might be additive for arbitrary quantum states.

In this paper, we resolve the above two open problems about Rains' bound by explicitly exhibiting a special class of two-qubit states whose Rains' bound and relative entropy of entanglement are known. We show that the Rains' bound is not additive and thereby the asymptotic (or regularized) Rains' bound will give a better upper bound on distillable entanglement. Meanwhile, the asymptotic relative entropy of entanglement (w.r.t. PPT states) of these two-qubit states is strictly smaller than the Rains' bound. Furthermore, an SDP lower bound $E_M$ for the asymptotic Rains' bound is introduced and it is the first computable lower bound for entanglement cost of general bipartite quantum states.
Meanwhile, this bound is proved to be the best known upper bound of the deterministic distillable entanglement, which gives the PPT-assisted asymptotic rate for some states, including all the pure states and the mixed states $\rho^{(\alpha)}(0<\alpha\le0.2)$ in Ref. \cite{Wang2016}.

Before we present our main results, let us first review some notations and preliminaries. In the following we will frequently use symbols such as $A$ (or $A'$) and $B$ (or $B'$) to denote (finite-dimensional) Hilbert spaces associated with Alice and Bob, respectively. The set of linear operators over $A$ is denoted by $\cL(A)$. Note that for a linear operator $R$ over a Hilbert space, we define $|R|=\sqrt{R^\dagger R}$, and the trace norm of $R$ is given by $\|R\|_1=\tr |R|$, where $R^\dagger$ is the conjugate transpose of $R$. The operator norm $\|R\|_\infty$ is defined as the maximum eigenvalue of $|R|$. A positive semidefinite operator $E_{AB} \in \cL(A\ox B)$ is
said to be PPT if $E_{AB}^{T_{B}}\geq 0$, i.e.,
$(\ketbra{i_Aj_B}{k_Al_B})^{T_{B}}=\ketbra{i_Al_B}{k_Aj_B}$.

The concise definition of entanglement of distillation by LOCC is given in Ref. \cite{Plenio2007} as follows:
$$E_D(\rho_{AB})=\sup\{r: \lim_{n \to \infty} [\inf_\Lambda  \|\Lambda(\rho_{AB}^{\ox n})- \Phi(2^{rn})\|_1]=0\},$$
where $\Lambda$ ranges over LOCC operations and $\Phi(d)=1/d\sum_{i,j=1}^d\ketbra{ii}{jj}$ represents the standard $d\otimes d$ maximally entangled state.  When $\Lambda$ ranges over PPT operations, the PPT-assisted distillable entanglement is defined by $E_{D,PPT}$.

The Rains' bound was introduced in Ref. \cite{Rains2001} and refined in Ref. \cite{Audenaert2002}  as a convex optimization problem  as follows: 
\begin{equation}\label{rains}
R(\rho)=\min S(\rho||\tau) \text{ s.t. } \ \tau\ge0, \tr |\tau^{T_B}|\le 1.
\end{equation}
In  this formula, $S(\rho||\sigma)=\tr (\rho\log \rho-\rho\log\sigma)$ denotes
the relative Von Neumann entropy. Rains' bound is important in entanglement theory and the generalized Rains information of a quantum channel is recently proved to be a strong converse rate for quantum communication \cite{Tomamichel2015a}.

The relative entropy of entanglement (REE)  \cite{Vedral1997, Vedral1998a, Vedral1997a} with respect to the PPT states is given by the following convex optimization problem:
\begin{equation}
E_{R,PPT}(\rho)= \min S(\rho || \sigma) \ \text{ s.t. }\  \sigma, \sigma^{T_B}\ge 0,\tr(\sigma)=1.
\end{equation}
For a general bipartite state $\rho$,  it holds that $E_{R,PPT}(\rho)\ge R(\rho)$. However, $E_{R,PPT}(\rho)$ equals to $R(\rho)$ for every two-qubit state $\rho$ \cite{Miranowicz2008} or the bipartite state with one qubit subsystem \cite{Girard2014}.  In particular, a two-qubit full-rank state $\sigma$ is the closest seperable state of any state $\rho$ in the following form \cite{Miranowicz2008, Friedland2011}:
\begin{equation}\label{CSS}
\rho=\sigma-xG(\sigma),
\end{equation} 
and 
\begin{equation}
G(\sigma)=\sum_{i,j}G_{i,j}\proj{v_i}(\proj \phi)^{T_B}\proj {v_j},
\end{equation}
with $\text{span}(\ket \phi)$ is the kernel (or null space) of $\sigma^{T_B}$ and
$G_{i,j}=\lambda_i$ when $\lambda_i=\lambda_j$ and $G_{i,j}=(\lambda_i-\lambda_j)/(\ln \lambda_i-\ln \lambda_j)$ when $\lambda_i\ne\lambda_j$, where $\lambda_i$ and $\ket {v_i}$ are the eigenvalues and eigenvectors of $\sigma$, respectively.

The asymptotic relative entropy of entanglement is given by
\begin{equation}
E_{R,PPT}^{\infty}(\rho)=\inf_{n\ge 1}  \frac{1}{n}E_{R,PPT}(\rho^{\ox n}).
\end{equation}

The numerical estimation of relative entropy of entanglement with respect to the PPT states is introduced in Refs. \cite{Zinchenko2010, Girard2015}, i,e, 
can be estimated by a Matlab program. Suppose that the estimation of $E_{R,PPT}(\rho)$ by  in Refs. \cite{Zinchenko2010, Girard2015} is $E_R^{+}(\rho)$, and the inequality $E_R^{+}(\rho)=S(\rho||\sigma)\ge E_{R,PPT}(\rho)$ holds 
since the algorithm indeed provides a feasible PPT state $\sigma$ which is almost optimal. This algorithm is implemented by CVX \cite{Grant2008} (a Matlab software for disciplined convex programming) and QETLAB \cite{NathanielJohnston2016}.
In low dimensions, this algorithm  provides an estimation  $E^+_R(\rho)$ with an absolute error smaller than $10^{-3}$, i.e.  $E_{R,PPT}(\rho)+10^{-3}\ge E^+_R(\rho)\ge E_{R,PPT}(\rho)$. 

The SDP upper bound on distillable entanglement $E_W(\rho)=\log W(\rho)$ for a bipartite state $\rho$ is introduced in Ref. \cite{Wang2016}, i.e.,
\begin{equation}\label{prime WN}
\begin{split}
W(\rho)= \max \ \tr  \rho R_{AB}, \  |R_{AB}^{T_{B}}| \le  \1, R_{AB}\ge0. 
\end{split}\end{equation}
Semidefinite programming (SDP) \cite{Vandenberghe1996} is a powerful tool in quantum information theory with many applications (e.g., \cite{Piani2015,Jain2011a,Wang2016a,Skrzypczyk2014,Wang2016g,Berta2015b,Kogias2015a,Wang2016f,Doherty2002a,Li2017}), which can be implemented by CVX \cite{Grant2008} and QETLAB \cite{NathanielJohnston2016}.

\section{Main Results}
\subsection{Nonadditivity of Rains' Bound} 
We first introduce a class of two-qubit states $\rho_r$ whose closest separable states can be derived by the result in Ref. \cite{Miranowicz2008}. Thus, the Rains' bound of $\rho_r$ is exactly given. Then we apply the algorithm in Refs. \cite{Zinchenko2010, Girard2015} to demonstrate the gap between $R(\rho^{\ox 2})$ and $R(\rho)$.
\begin{theorem}
There exists a two-qubit state $\rho$ such that $$R(\rho^{ \ox 2})<2R(\rho).$$
Meanwhile,
$$E^{\infty}_R(\rho)< R(\rho).$$
\end{theorem}
\begin{proof}
Firstly, we construct two-qubit states $\rho_r$ and $\sigma_r$ satisfying Eq. (\ref{CSS}). Then we have $R(\rho_r)=S(\rho_r || \sigma_r)$.
Suppose that
\begin{align*}
\sigma_r&=\frac{1}{4}\proj{00}+\frac{1}{8}\proj{11}+r\proj{01}\\
&+(\frac{5}{8}-r)\proj{10}+\frac{1}{4\sqrt2}(\ketbra{01}{10}+\ketbra{10}{01}).
\end{align*}
The positivity of $\sigma_r$ requires that  $\frac{5-\sqrt 17}{16}\le r\le \frac{5+\sqrt 17}{16}$. Assume that $r\ge 5/8-r$ and we can further choose $0.3125\le r \le 0.57$ for simplicity.

Meanwhile, let us choose
\begin{align*}
\rho_{r}
=&\frac{1}{8}\proj{00}+x\proj{01}+\frac{7-8x}{8}\proj{10}\\
&+\frac{32r^2-(6+32x)r+10x+1}{4\sqrt 2}(\ketbra{01}{10}+\ketbra{10}{01})
\end{align*}
with 
$$x=r+\frac{32r^2-10r+1}{256r^2-160r+33}+\frac{(16r-5)y^{-1}}{32\ln{(5/8-y)}-32\ln{ ({5}/8+y)}}$$
and $y=(4r^2-{5r}/{2}+{33}/{64})^{1/2}$. It is clear that $\tr \rho_r=1$ and we set $0.3125\le r \le 0.5480$ to ensure the positivity of $\rho_r$.

One can readily verify that $\rho_{r}=\sigma_r-{3}G(\sigma_r)/2$.
Therefore, $\sigma_r$ is the closest separable state (CSS) for $\rho_r$ and we have that
\begin{equation}
R(\rho_r)=E_{R,PPT}(\rho_r)=S(\rho_r || \sigma_r).
\end{equation}

In particular, let us first choose $r_0=0.547$, the Rains' bound of $\rho_{r_0}$ is given by
 $$R(\rho_{r_0})=E_{R,PPT}(\rho_{r_0})=S(\rho_{r_0} || \sigma_{r_0})\simeq 0.3891999.$$
Furthermore, applying the algorithm in Refs. \cite{Zinchenko2010, Girard2015},
we can find a PPT state $\sigma_0$ such that  $$ E_R^+(\rho_{r_0} ^{ \ox 2} )= S(\rho_{r_0} ^{ \ox 2}  || \sigma_0)\simeq 0.7683307.$$
The numerical value of relative entropy here is calculated based on the Matlab function ``logm''  \cite{logm} and
the function ``Entropy'' in QETLAB \cite{NathanielJohnston2016}. In this case, the accuracy is guaranteed by the fact  $\|e^{\text{logm}(\sigma_{r_0})}-\sigma_{r_0}\|_1\le 10^{-16}$ and $\|e^{\text{logm}(\sigma_{0})}-\sigma_0\|_1\le 10^{-14}$.
Noting that the difference between $2R(\rho_{r_0})$ and $ E_R^+(\rho_{r_0} ^{ \ox 2} )$ is already $1.00691\times 10^{-2}$,
we have that
\begin{align*}
R(\rho_{r_0}^{ \ox 2} ) \le E_{R,PPT}(\rho_{r_0}^{ \ox 2} )\le E_R^+(\rho_{r_0} ^{ \ox 2} )< 2R(\rho_{r_0}).
\end{align*}
It is also easy to observe that
$$E^{\infty}_{R,PPT}(\rho_{r_0}) \le \frac{1}{2}E_{R,PPT}(\rho_{r_0}^{ \ox 2} )<R(\rho_{r_0}).$$

When $0.45\le r\le 0.548$, we show the gap between $2R(\rho_r)$ and $E_R^+(\rho_r^{ \ox 2} )$ in  FIG. 1.  
\end{proof}
\begin{figure}[htbp]
  \centering
\includegraphics[width=0.42\textwidth]{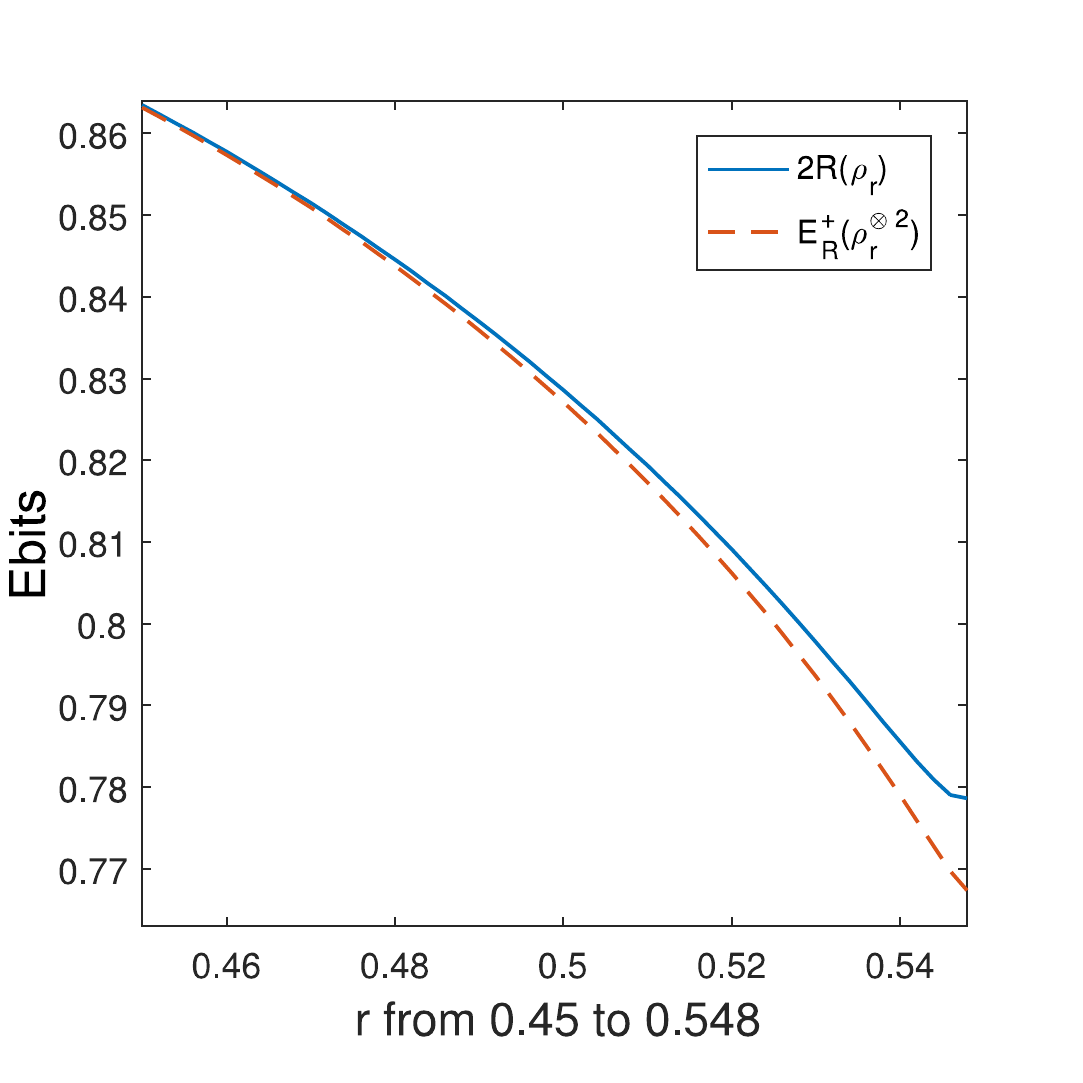}
\caption{This plot demonstrates the difference  between $2R(\rho_r)$ and $E_R^+(\rho_r^{ \ox 2} )$ for $0.45\le r \le 0.548$. The dashed line depicts  $E_R^+(\rho_r^{ \ox 2} )$ while the solid line depicts  $2R(\rho_r)$.}
\end{figure}

Since Rains' bound is not additive,  the asymptotic Rains' bound \cite{Hayashi2017b} can provide better upper bound on the distillable entanglement, i.e.,
\begin{equation}
E_{D,PPT}(\rho) \le R^{\infty}(\rho)=\inf_{n \ge 1}  \frac{1}{n}R(\rho^{\ox n})\le R(\rho),
\end{equation}
and the last inequality can be strict.

%%%%%%%%%%%%%%%%%%%%%%%%%%%%%%%%%%%%%%%%%%
\subsection{A SDP lower bound for entanglement cost}
Since computing the entanglement cost of a bipartite state is very difficult, we introduce an efficiently computable lower bound to evaluate the  entanglement cost.  

For a bipartite quantum state $\rho$, we introduce
\begin{equation}\label{prime 2 M}\begin{split}
E_M(\rho)=-\log M(\rho)=& -\log\max \tr P_{AB}V_{AB}, \\
&\text{ s.t. } \tr |V_{AB}^{T_B}|=1, V_{AB}\ge 0, 
\end{split}\end{equation}
where $P_{AB}$ is the projection onto the support of $\rho$.
And $M(\rho)$ is also given by the following SDP:
\begin{equation}\label{prime 1 M}\begin{split}
 M(\rho)= \max& \tr P_{AB} Z_{AB}, \\
\text{ s.t. } &\tr (X_{AB}+Y_{AB})=1, \\
&Z_{AB}\le (X_{AB}-Y_{AB})^{T_B}\\
&X_{AB},Y_{AB},Z_{AB}\ge 0, 
\end{split}\end{equation}
And its dual SDP is given by
\begin{equation}\label{dual 2 M}
M(\rho)= \min   \|R_{AB}^{T_B}\|_\infty, \text{ s.t. }   R_{AB}\ge P_{AB}.
\end{equation}
The optimal values of the primal and the dual SDPs above coincide by strong duality, which can be proved by Slater’s theorem.

For any two bipartite states $\rho_{AB} $ and $\sigma_{A'B'}$, by utilizing semidefinite programming duality,
it is not difficult to prove that $$E_M(\rho_{AB} \ox\sigma_{A'B'})=E_M(\rho_{AB} )+E_M(\sigma_{A'B'}).$$

Furthermore, for any state bipartite $\rho$,
$E_M(\rho)=0$ if and only if $\supp(\rho)$ contains the support of a PPT state $\sigma$, i.e. $\supp(\sigma)\subseteq\supp(\rho) $. Too see this,
if there exists PPT state $\sigma$ such that $\supp(\sigma)\subseteq\supp(\rho) $, then $E_M(\rho)=0$. On the other hand, if any state $\sigma$ satisfies $\supp(\sigma)\subseteq\supp(\rho) $ is NPPT. Let the optimal solution to SDP (\ref{prime 2 M}) be $V$, where $V\ge0$ and $\tr|V^{T_B}|=1$. It is clear that $\tr V\le 1$. Thus, we have $\tr V=1$ when $E_M(\rho)=0$.  Hence, $V$ is a PPT state and $\supp(V)\subseteq\supp(\rho) $. This leads to a contradiction.

\begin{theorem}\label{rains lower bound}
For any bipartite state $\rho$,
$$E_M(\rho)\le R^{\infty}(\rho)\le E_C(\rho).$$
\end{theorem}
\begin{proof}
Suppose that the optimal solution to Eq. (\ref{rains}) of $R(\rho)$ is $V$, then $V\ge0$ and $\tr|V^{T_B}|=t\le1$. Thus, $V/t$ is a feasible solution to SDP (\ref{prime 2 M}) of $M(\rho)$, this means that $E_M\le -\log \tr PV/t =-\log \tr PV+\log t$, where $P$ is the projection onto $\supp(\rho)$.

On the other hand, let $\cN(\sigma)=P\sigma P+(\1-P)\sigma(\1-P)$, 
then by the monotonicity of quantum relative entropy, 
\begin{equation}\begin{split}
S(\rho||V)&\ge S(\cN(\rho)||\cN(V))=S(\rho||PVP)\\
                &=S(\rho||\frac{PVP}{\tr PVP})-\log\tr PV\\
                &\ge -\log\tr PV \ge E_M(\rho).
\end{split}\end{equation}
Noting that $E_M(\cdot)$ is additive, we have that 
$$E_M(\rho)\le \inf_{n \ge 1} \frac{1}{n} R(\rho^{\ox n})=R^{\infty}(\rho).$$

Finally, it is clear that
$$E_M(\rho)  \le  R^{\infty}(\rho) \le E_{R,PPT}^{\infty}(\rho) \le   E_{C}(\rho),$$
where the last inequality is from Ref. \cite{Hayashi2017b}.
\end{proof}
\begin{remark}
As an application of this lower bound,  one can also give an SDP lower bound for \emph{the entanglement cost of quantum channels} \cite{Berta2013}, i.e. the rate of entanglement (ebits) needed to asymptotically simulate a quantum channel $\cN$ with free classical communication.  
\end{remark}

%%%%%%%%%%%%%%%%%%%%%%%%%%%%%%%%%%%%%%%%%%%%%%%%%%%%%
\subsection{Deterministic distillable entanglement}
In this section, we show that $E_M$ is the best upper bound on the deterministic distillable entanglement of bipartite states. The bipartite pure state case is completely solved in Refs. \cite{Matthews2008, Duan2005b}. For a general state, the PPT-assisted deterministic distillation rates depend only on the support of this state \cite{Wang2016}. Note that the support $\supp(\rho)$ of a state $\rho$ is defined to be the space spanned by the eigenvectors with non-zero eigenvalues of $\rho$. 
The exact value of the one-copy PPT-assisted deterministic distillation rate of a given bipartite state $\rho$ is $E_{0,D,PPT}^{(1)}(\rho)= -\log W_0(\rho)$ \cite{Wang2016}, where ${W_0}(\rho)$ is given by
\begin{equation}\label{prime hat W0}
{W_0}(\rho)=  \min  \|R^{T_B}\|_\infty, 
\text{ s.t. }   P_{AB}\le R\le \1_{AB}.
\end{equation}
Here, $P_{AB}$ is the projection onto $\supp(\rho)$.

\begin{theorem}\label{estimation}
For any bipartite state $\rho$,
$$E_{0,D,PPT}(\rho) \le E_M(\rho) \le E_W(\rho).$$
\end{theorem}
\begin{proof}
To prove $E_{0,D,PPT}(\rho) \le -\log M(\rho)$, suppose that the optimal solution to SDP (\ref{prime hat W0}) of 
$W_0(\rho)$ is $R_0$. It is clear that $R_0$ is also a feasible solution to SDP (\ref{dual 2 M}) of $M(\rho)$. 
Thus, $W_0(\rho)= \|{R_0}^{T_B}\|_\infty\ge M(\rho)$.
Furthermore, ${W_0}(\rho^{\ox n})\ge M(\rho^{\ox n})= M(\rho)^n$.

Hence,
\begin{align*}
E_{0,D,PPT}(\rho) &=\lim_{n\to \infty}  -\frac{1}{n}\log {W_0} (\rho^{\otimes n})\\
&\le \lim_{n\to \infty}  -  \frac{1}{n}\log M(\rho)^n= E_M(\rho).
\end{align*}

Finally, to prove $E_M(\rho) \le E_W(\rho)$, suppose that the optimal solution to SDP (\ref{dual 2 M}) is $R$, then we have $R\ge P \ge 0$. Let $R'=R/\|R^{T_B}\|_\infty $ and it is easy to see the positivity of $R'$ and the fact that $|R'^{T_B}|\le \1$, which means that $R'$ is a feasible solution to SDP (\ref{prime WN}). Therefore, 
$E_W(\rho)\ge \log \tr \rho R'
 \ge  \log \tr \rho P/\|R^{T_B}\|_\infty  = -\log \|R^{T_B}\|_\infty  =E_M(\rho)$. 
\end{proof}
\begin{remark}
For any bipartite state $\rho$, if the support of $\rho$ contains a PPT state $\sigma$, then  $E_M(\rho)=0$ and we have that $E_{0,D,PPT}(\rho) =0$.
Thus $\rho$ is bound entanglement for exact distillation under both LOCC or PPT operations.
\end{remark}

We further show the estimation of Theorem \ref{estimation} in Fig.2 by a class of $3\otimes 3$ states in Ref. \cite{Wang2016} defined by 
$$\rho^{(\a)}=\frac{1}{3}\sum_{m=0}^{2} (X^\dagger \ox X)^m\proj{\psi_0} (X\ox X^\dagger )^m,$$
 where $\ket{\psi_0}=\sqrt{\a}\ket{00}+\sqrt{1-\a}\ket{11} (0<\a \le 0.5)$ and 
 $X=\sum_{j=0}^{2}\ketbra {j\oplus 1}{j}$. An interesting fact is that $E_M(\rho^{(\a)})$ is tight for $E_{0,D,PPT}(\rho^{(\a)})$ when $0<\a\le1/5$, which is proved in the following Proposition.
 \begin{figure}[ht]
\vspace{-0.3cm}
\includegraphics[width=0.32\textwidth]{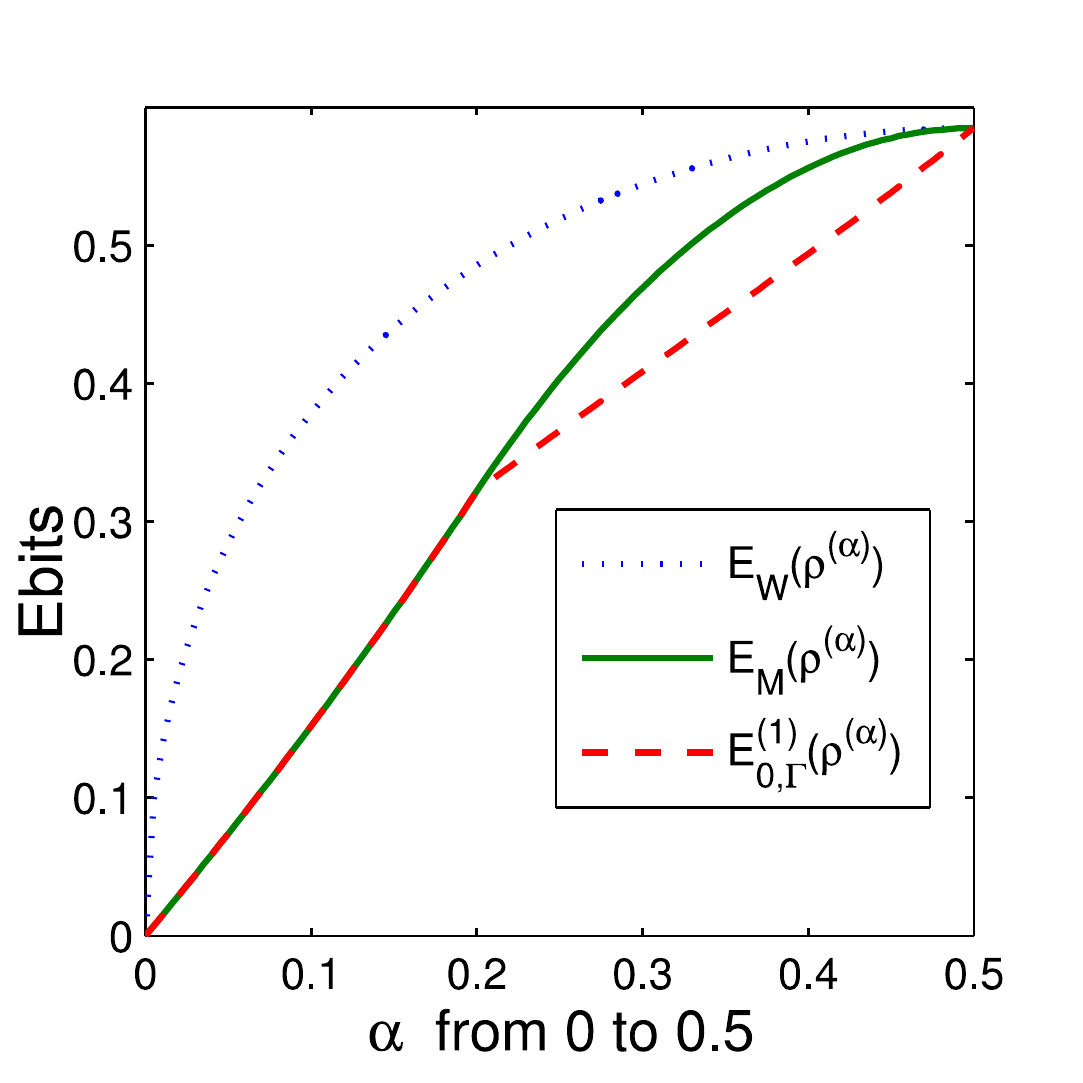}
 \caption{This plot presents the estimation of $E_{D,PPT}(\rho^{(\a)})$ and $E_{0,D,PPT}(\rho^{(\a)})$. The dot line depicts $E_W(\rho^{(\a)})$,  the dash line depicts $E_{0,D,PPT}^{(1)}(\rho^{(\a)})$ and the solid line depicts $E_M(\rho^{(\a)})$.}
\end{figure}

\begin{proposition}\label{tight case}
For any bipartite state $\rho$ with support projection $P$, suppose that the eigenvector $\ket{\psi}$ of $P^{T_B}$ with the eigenvalue $\|P^{T_B}\|_\infty$ is a product state, then 
$$E_{0,D,PPT}(\rho)=E_{M}(\rho)=-\log\|P^{T_B}\|_\infty\le E_{D,PPT}(\rho).$$
\end{proposition}
\begin{proof}
In Ref. \cite{Wang2016}, it shows that $E_{0,D,PPT}(\rho)\ge-\log\|P^{T_B}\|_\infty$. If $\proj{\psi}$ is PPT, then we can choose $V=\proj{\psi}$  and it is easy to see $V$ is a feasible solution to SDP (\ref{prime 2 M}) of $M(\rho)$. Thus, $E_M(\rho)\le -\log \tr P^{T_B}\proj{\psi}=-\log\|P^{T_B}\|_\infty$.
\end{proof}

For any pure state $\proj{\phi}$, 
suppose that $\ket{\phi}$ has the Schmidt decomposition $\ket \phi=\sum_{i=1}^{m}\lambda_i\ket{ii}$ with $\lambda_1^2\ge...\ge\lambda_m^2$ and $\sum_{i=1}^m\lambda_i^2=1$.
Then $\proj{\phi}^{T_B}=\sum_{i=1}^m\lambda_i^2\proj{ii}+\sum_{i\ne j}\lambda_i\lambda_j\ket{ji}\bra{ij}$. Thus, $\|P^{T_B}\|_\infty=\lambda_1^2$ and the corresponding eigenvector is $\proj{11}$. Hence, by Proposition \ref{tight case}, $E_{0,D,PPT}(\proj{\phi})=E_{M}(\proj{\phi})=-\log\|\proj{\phi}^{T_B}\|_\infty$. This rate can be achieved by LOCC \cite{Duan2005b}.

\begin{example}
For the $\rho^{(\a)}$, when $0<\a\le1/5$,
we have that $E_{0,D,PPT}(\rho^{(\a)} )=E_{M}(\rho^{(\a)} )=-\log (1-\a)$. 
Let $U=X^\dagger \ox X$, the projection onto $\supp(\rho^{(\a)})$ is $P_\a=\sum_{m=0}^{2} U^m\proj{\psi_0} (U^\dagger)^m$. Therefore, 
\begin{align*}
P_\a^{T_B}=&2\sqrt{\a(1-\a)}\proj{v_1}-\sqrt{\a(1-\a)}(\proj{v_2}+\proj{v_3}) \\
&+\sum_{m=0}^{2} U^m[(1-\a)\proj{11} +\a\proj{00}](U^\dagger)^m,
\end{align*}
 where $\ket{v_1}=\frac{1}{\sqrt3}(\ket{01}+\ket{10}+\ket{22})$,  $\ket{v_2}=\frac{1}{\sqrt6}\ket{01}+\frac{1}{\sqrt6}\ket{10}-\sqrt{\frac{2}{3}}\ket{22})$ and $\ket{v_3}=\frac{1}{\sqrt2}(\ket {01} -\ket{10})$.
When $0<\a\le1/5$, we always have $1-\a\ge 2\sqrt{\a(1-\a)}$. Therefore, $\|P_\a^{T_B}\|_\infty=1-\a$ and the corresponding eigenvector is $\proj{11}$. Applying Proposition \ref{tight case}, the proof is done.
\end{example}

%%%%%%%%%%%%%%%%%%%%%%%%%%%%%%%%%%%%%%%%%%%%%%%%%%%

%%%%%%%%%%%%%%%%%%%%%%%%%%%%%%%%%%%%%%%%%%%%%%%%%%%
\section{Conclusions and Discussions}
We show that the Rains' bound is neither additive nor equal to the asymptotic relative entropy of entanglement by explicitly constructing a special class of mixed two-qubit states. We also show an SDP lower bound $E_M$ for the asymptotic Rains' bound. These results solve two open problems in entanglement theory and provide an efficiently computable lower bound for the entanglement cost of general bipartite states. This bound also has desirable properties such as additivity under tensor product and vanishing if and only if the support of the given state contains some PPT state. We further show that $E_M$ is the best upper bound for the deterministic distillable entanglement, which also gives the PPT-assisted deterministic distillation rate in some conditions, including all the pure states and some classes of the mixed states. 

It is of great interest to determine whether the asymptotic Rains' bound and the PPT distillable entanglement always coincide. It would also be interesting to decide whether $E_{0,D,PPT}(\rho)= E_M(\rho)$ for any bipartite state $\rho$ and to study the relationship between $E_M$ and the newly established SDP lower bound of the PPT-assisted entanglement cost \cite{Wang2016d}.
%Recently, another lower bound of entanglement cost has been establised in \cite{Wang2016d}, it is 
% and further compare Rains' bound to $E_W$.
%Relations between some fundamental entanglement measures and related quantities are showed in Fig. 3. We hope these new findings may be useful in studying other problems in quantum information theory.

%\begin{figure}[htbp]
%  %\centering
%\includegraphics[width=0.19\textwidth]{zoo_E4.jpg}
% \caption{This plot presents relations between some fundamental entanglement measures and related quantities. $E_A\to E_B$ indicates that $E_A(\rho) \ge $$E_B(\rho)$ for all $\rho$. The gap between $E_R^{\infty}$ and $R^{\infty}$ is recently proved in Ref. \cite{Wang2016d}.}
%\end{figure}

We were grateful to A. Winter, M. Tomamichel and M. B. Plenio for helpful suggestions. This work was partly supported by the Australian Research Council (Grant Nos. DP120103776 and FT120100449).

% %\bibliographystyle{IEEEtran}
% \bibliographystyle{apsrev4-1}
% % argument is your BibTeX string definitions and bibliography database(s)
% \bibliography{Bib}

\end{document}